\documentclass[superscriptaddress,pra,twocolumn,nofootinbib,notitlepage, floatfix]{revtex4-1}
\usepackage{amssymb}
\usepackage{dsfont}
\usepackage{complexity}
\usepackage[utf8]{inputenc} 
\usepackage[T1]{fontenc}    
\usepackage{hyperref}       
\usepackage{url}            
\usepackage{booktabs}       
\usepackage{amsfonts}       
\usepackage{nicefrac}       
\usepackage{microtype}      
\usepackage{lipsum}
\usepackage{amsmath}
\usepackage{physics}
\usepackage{placeins}
\usepackage{bbold}
\usepackage{tikz}
\usepackage{forest}
\usepackage{multirow}
\usetikzlibrary{calc,arrows.meta,positioning, shapes.geometric, arrows, shapes.symbols,shapes.callouts,patterns}

\tikzset{
    every node/.style={font=\sffamily\small},
    main node/.style={thick,circle ,draw},
    visible node/.style={thick,rectangle ,draw}
}
\usepackage{qcircuit}
\usepackage{graphicx}
\usepackage{graphics}
\usepackage{amsthm}
\usepackage[caption=false]{subfig}
\usepackage{latexsym}
\usepackage{floatrow}
\usepackage{mathtools}
\usepackage{verbatim}
\usepackage{upgreek}
\usepackage{textgreek}
\usepackage{thmtools}
\usepackage{thm-restate}

\usepackage{hyperref}

\usepackage{cleveref}

\newtheorem{theorem}{Theorem}

\newtheorem{proposition}[theorem]{Proposition}
\newtheorem{lemma}[theorem]{Lemma}

\DeclareMathOperator{\EX}{\mathbb{E}}

\begin{document}

\title{Entanglement Induced Barren Plateaus}
\author{Carlos Ortiz Marrero}
\affiliation{Data Sciences and Analytics Group,
  Pacific Northwest National Laboratory,
  Richland, WA 99354 }
\email{carlos.ortizmarrero@pnnl.gov}

\author{M\'aria Kieferov\'a}
\affiliation{  Centre for Quantum Computation and Communication Technology,
Centre for Quantum Software and Information,
University of Technology Sydney,
NSW 2007, Australia}
\email{maria.kieferova@uts.edu.au}

\author{Nathan Wiebe}
\affiliation{ Department of Computer Science, University of Toronto, ON M5S 1A1, Canada}
\email{nwiebe@cs.toronto.edu}

\date{\today}

\begin{abstract}

We argue that an excess in entanglement between the visible and hidden units in a Quantum Neural Network can hinder learning. In particular, we show that quantum neural networks that satisfy a volume-law in the entanglement entropy will give rise to models not suitable for learning with high probability. Using arguments from quantum thermodynamics, we then show that this volume law is typical and that there exists a barren plateau in the optimization landscape due to entanglement.  More precisely, we show that for any bounded objective function on the visible layers, the Lipshitz constants of the expectation value of that objective function will scale inversely with the dimension of the hidden-subsystem with high probability. We show how this can cause both gradient descent and gradient-free methods to fail. We note that similar problems can happen with quantum Boltzmann machines, although stronger assumptions on the coupling between the hidden/visible subspaces are necessary. We highlight how pretraining such generative models may provide a way to navigate these barren plateaus.

\end{abstract}

\maketitle

\section{Introduction}
In recent years the prospects of quantum machine learning (QML) and quantum deep neural network have gained notoriety in the scientific community. QML builds on the success of traditional machine learning and the potential for quantum speedup. The QML field has enjoyed increased attention for quantum algorithms for principal component analysis~\cite{lloyd2014quantum}, support vector machines~\cite{rebentrost2014quantum}, kernel methods~\cite{schuld2019quantum, havlivcek2019supervised}, and quantum neural networks (QNN)~\cite{schuld2014quest,wiebe2016quantum, farhi2018classification,schuld2019quantum,wiebe2019generative}  but experiences setbacks in the form of dequantization techniques~\cite{tang2019quantum, tang2018quantum, gilyen2018quantum}.

A key part of a successful quantum machine learning algorithm is an efficient training algorithm. In recent years, several barren plateau results~\cite{mcclean2018barren, cerezo2020cost, sharma2020trainability, wang2020noise,holmes2020barren} put limitations on the gradient-based training of QNNs. Our result complements the growing literature on barren plateaus in quantum computing. McClean et al.~\cite{mcclean2018barren} first showed that unitary quantum neural networks generically suffer from vanishing gradients exponentially in the number of qubits. This issue stems from the concentration of measure~\cite{bremner_random, PhysRevLett.102.190501} and was subsequently demonstrated for other QNNs~\cite{cerezo2020cost, sharma2020trainability}. Another type of a barren plateau emerges from hardware noise in the system~\cite{wang2020noise}. The key observation that we put forward in this work is that barren plateaus can occur because of an excess of entanglement in deep quantum models.

In this paper, we prove that entanglement between visible and hidden units hinders the learning process. 
Inclusion of hidden units is essential in traditional machine learning. Without them, the expressive power of neural networks would be severely limited and deep learning all but impossible. In spite of this, there has been very little attention paid to the effect of hidden units on the training of QNNs. Surely, the expressive power of hidden units would translate to the quantum world? Numerical experiments seem to contradict this intuition. A small scale numerical study~\cite{kieferova2017tomography} showed that the inclusion of hidden units to quantum Boltzmann machines did not lead to a higher quality of reproduction. While this could be explained due to the small size of the QNN and simple data, in our work we show that quantum Boltzmann machines do not benefit from a large number of hidden units. 

We build intuition from exploring the statistical relationship between a random state and maximally entangled states in a bipartite quantum system.  A classic thermalization result~\cite{popescu2006entanglement} shows that for a random initial state, the state on the visible units is with high probability exponentially close to a maximally mixed state. However, if the state is chosen from a k-design, its distance to a maximally mixed state is bounded by a polynomial in $k$~\cite{low2009large}. 
We show that it is very difficult to escape from this state because the gradients will be exponentially small. As such, for a wide array of QNNs, randomness and entanglement hinder the training.

This surplus of entanglement to some extent defeats the purpose of deep learning by causing information to be non-locally stored in the correlations between the layers rather than in the layers themselves.  As a result, when one tries to remove the hidden units, as is customary in deep learning, we find that the resulting state is close to the maximally mixed state.  Indeed, we show that such situations are generic as well and that gradient descent methods are unlikely to allow the user to escape from such a plateau at low cost.  This observation holds for both "feedforward" QNNs as well as Boltzmann machines and suggests that if quantum effects are to be used to improve classical models then they must be used surgically.  Furthermore, our work establishes a link between the thermalization literature and quantum machine learning that has been hitherto absent from the literature.

We focus on two types of QNNs depicted in Figure~\ref{fig:models}, feed-forward unitary quantum neural networks inspired by QAOA, and Quantum Boltzmann machines~\cite{amin2018quantum, kieferova2017tomography, wiebe2019generative}.

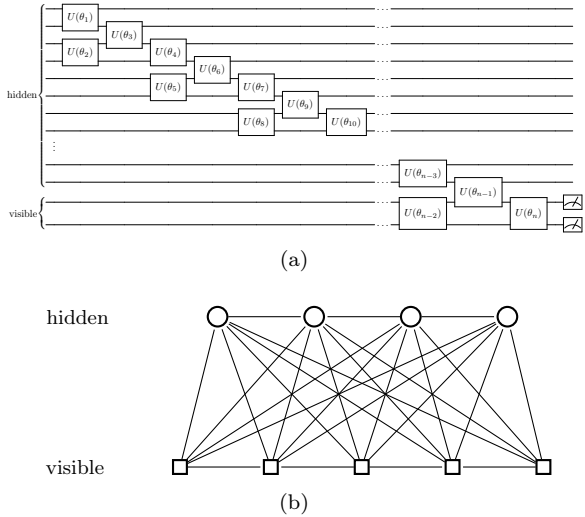
\begin{figure}[tb]
  \subfloat[]{%
    \includegraphics[width=0.9\textwidth]{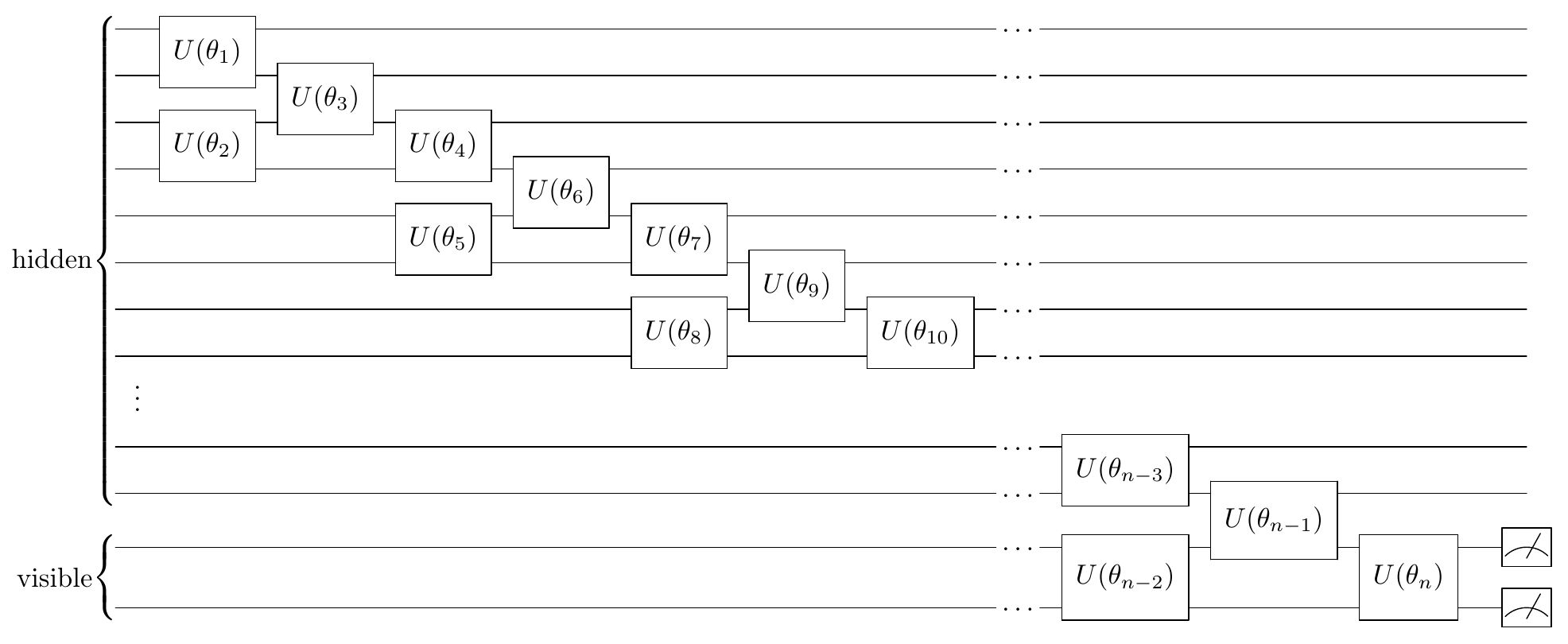}\label{deep_net}%
  }\\
  
  \subfloat[ ]{
      \begin{tikzpicture}[-,>={Stealth[round,sep]},shorten >=1pt,auto,node distance=1 cm, every node/.style={scale=0.85}]

    \node[main node] (h1) at (0.5,0)  {};
    \node[main node] (h2) [right =of h1] {};
    \node[main node] (h3) [right =of h2] {};
    \node[main node] (h4) [right =of h3] {};
    
    \node[text width=3cm] at (-0.5, 0) 
    {hidden};
    
    \node[text width=3cm] at (-0.5,-2) 
    {visible};

    \node[visible node] (v1) at (0, -2 ) {};
    \node[visible node] (v2) [right =of v1] {};
    \node[visible node] (v3) [right =of v2]{};
    \node[visible node] (v4) [right =of v3]{};
    \node[visible node] (v5) [right =of v4]{};

    \draw (h1) -- (h2);
    \draw (h3) -- (h2);
    \draw (h3) -- (h4);
    
    \draw (v1) -- (v2);
    \draw (v3) -- (v2);
    \draw (v3) -- (v4);
    \draw (v5) -- (v4);
    
    \draw (v1) -- (h1);
    \draw (v1) -- (h2);
    \draw (v1) -- (h3);
    \draw (v1) -- (h4);
    
    \draw (v2) -- (h1);
    \draw (v2) -- (h2);
    \draw (v2) -- (h3);
    \draw (v2) -- (h4);
    
    \draw (v3) -- (h1);
    \draw (v3) -- (h2);
    \draw (v3) -- (h3);
    \draw (v3) -- (h4);
    
    \draw (v4) -- (h1);
    \draw (v4) -- (h2);
    \draw (v4) -- (h3);
    \draw (v4) -- (h4);
    
    \draw (v5) -- (h1);
    \draw (v5) -- (h2);
    \draw (v5) -- (h3);
    \draw (v5) -- (h4);

\end{tikzpicture}\label{qbm}
}
  \caption[]{Examples of QNNs. \subref{deep_net}~
    A quantum unitary network characterized by a  circuit with  parameterized unitaries $U_j = e^{-i \theta_j H_j}$ where $\theta_j$ are the parameters we aim to learn and $H_j$ Hamiltonians that specify the QNN.  
The output is then $U(\theta_1, \dots, \theta_n)\ket{\psi_0}$ where $\ket{\psi_0}$ can be taken to be  $\ket{0\dots 0}$ for generative learning. 
In this model, \emph{ visible units} correspond to the qubits on which we evaluate the objective function, in this case the last two registers. The remaining qubits are called \emph{hidden units}. 
\subref{qbm}~ Quantum Boltzmann machines defined on a graph. Each edge and each vertex correspond to a weight on a local Hamiltonian corresponding to the pair of qubits or a single qubit.  The top layer of units (circles) corresponds to visible units and the bottom layer (rectangles) are hidden units.
QBMs models data as a thermal state 
$\frac{e^{-H(\theta)}}{Z(\theta)}:= \frac{e^{-\sum_i \theta_i H_i}}{\text{Tr}\left(e^{-\sum_i \theta_i H_i} \right)}.$ 
Without loss of generality, we will take ${\rm Tr}(H) = 0$ for all quantum Boltzmann machines.  The aim when training a quantum Boltzmann machine is to learn a vector $\theta$ such that for a training objective function given by $O_{\rm obj}$ that acts on the visible subsystem, we maximize ${\rm Tr}(O_{\rm obj} {\rm Tr}_h(e^{-H(\theta)}/Z(\theta))$.
}\label{fig:models}
\end{figure}

Quantum Boltzmann machines can also be trained generatively~\cite{kieferova2017tomography}, meaning that rather than optimizing a training objective function that is a linear function of the density operator such as ${\rm Tr}( O_{\rm Obj} {\rm Tr}_h (e^{-H}/Z))$, we aim to optimize a non-linear function of the density operator such as the quantum relative entropy, i.e.  $S(\rho_{\rm train}||\rho(\theta))= \rm{Tr}(\rho_{\rm train} \log(\rho_{\rm train}) - \rho_{\rm train} \log(\rho(\theta)))$, by generating a quantum state $\rho(\theta)$ using the Boltzmann machine that optimizes this divergence with the training density operator.

\section{The impact of Entanglement on Deep Models}
The central question our work is to understand how the entanglement in the neural network affects the visible units. Instead of providing speedup, entanglement between visible and hidden units causes thermalization on the visible subsystem.  Thus, the inclusion of entanglement between the hidden and visible layers of a QNN can be, unless carefully controlled, harmful to the neural network model.

The relationship between the representational power of a neural network and the degree of entanglement between the visible and hidden systems was first discussed in~\cite{sarma2019machine}; however, here we re-examine this question and arrive at a different conclusion.  Specifically, we conclude that large amounts of entanglement (as quantified by a volume law) can be catastrophic for the model; whereas an area law scaling for the entanglement entropy between the hidden and visible can often be tolerated.

To see this, we need to make a few formal definitions.
Let $\mathcal{S}\subset \mathbb{C}^{D_vD_h \times D_vD_h}$ be a family of parameterized density operators where $D_v=2^{n_v}$ is the dimension of the $n_v$-qubit visible subspace and $D_h=2^{n_h}$ the dimension of the hidden subspace.  For each $\rho \in \mathcal{S}$, the qubits can be uniquely assigned to the vertices of a graph $\mathcal{G}$ on a vertex set $\mathcal{V}_h \bigcup \mathcal{V}_v$ where $\mathcal{V}_v$ consists of $\log_2(D_v)=n_v$ and $\mathcal{V}_h$ of $\log_2(D_h)=n_h$ qubits.  We then will define $n_v^{(j)}$ to be the number of vertices in $\mathcal{V}_v$ that are at least graph distance $j$ away from the vertices in $\mathcal{V}_h$ and define $n_h^{(j)}$ to be the analogous number for the vertices in $\mathcal{V}_h$.  We then say that $\mathcal{S}$ satisfies an area law if for all $\rho\in \mathcal{S}$, $S({\rm Tr}_h(\rho)) \in \Theta(n_v^{(1)})$ similarly we say that $\mathcal{S}$ satisfies a volume law if $|S({\rm Tr}_h(\rho)) - \log(D_v)|\in \Theta(D_v/D_h)$ where $S$ is the von Neumann entropy. With these definitions in place we can concisely claim our result.

\begin{proposition}  Let $\mathcal{S}$ be a family of density operators with visible dimension $D_v$ and hidden dimension $D_h \ge D_v$ with $n_v^{(1)}$ and $n_h^{(1)}$ qubits in the first visible and hidden layers respectively.  We then have that for any operator on the visible sub-system $O_{\rm obj}$ and $\rho \in \mathcal{S} \subset \mathbb{C}^{D_vD_h \times D_vD_h}$,
\begin{equation}
    |{\rm Tr}((O_{\rm obj} \otimes I) (\rho -I/(D_vD_h)))| 
\end{equation}
is in  $ O\Big(\|O_{\rm Obj}\|_\infty \sqrt{\log(D_v) - n^{(1)}_v }\Big)$
 if $\mathcal{S}$ satisfies an area law and is in $O\Big(\|O_{\rm Obj}\|_\infty \sqrt{(D_v/D_h)}\Big)$ if $\mathcal{S}$ satisfies a volume law.
\end{proposition}
\begin{proof}
The proof follows from standard inequalities for the quantum relative entropy 
\begin{align}
    \frac{1}{2}\|{\rm Tr}_h(\rho) - I/D_v\|_1^2&\le S({\rm Tr}_h(\rho)|| I/D_v) \nonumber\\
    &= -S({\rm Tr}_h(\rho))+ \log(D_v).\label{eq:trIneq}
\end{align}
Then from the von-Neumann trace inequality we have
\begin{align}
    &|{\rm Tr}((O_{\rm Obj} \otimes I) (\rho - I/D))|\le \|O_{\rm Obj}\|_\infty \|{\rm Tr}_h (\rho) - I/D_v\|_1 \nonumber\\
    &\le \|O_{\rm Obj}\|_\infty\sqrt{2(\log(D_v)-S({\rm Tr}_h(\rho)))}
\end{align}
If $\rho$ satisfied an area-law scaling then  $S({\rm Tr}_h(\rho)) \in \Theta( n_v^{(1)} ) $.  From which the claimed result for the area law scaling immediately follows.  If instead we assume that $\rho$ obeys a volume-law then $|S({\rm Tr}_h(\rho)) - \log(D_v)| \in \Theta(D_v/D_h) = \Theta(2^{n_v-n_h})$
\end{proof}

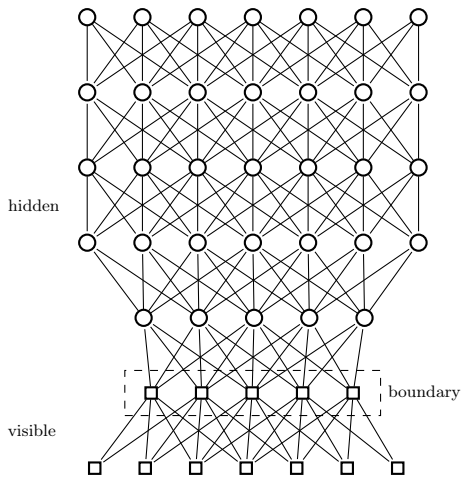
\begin{figure}[tb]
\begin{center}
    
\begin{tikzpicture}[-,>={Stealth[round,sep]},shorten >=1pt,auto,node distance=0.5 cm, every node/.style={scale=0.7}]

    \node[main node] (h1) at (-0.5,0)  {};
    \node[main node] (h2) [right =of h1] {};
    \node[main node] (h3) [right =of h2] {};
    \node[main node] (h4) [right =of h3] {};
    \node[main node] (h5) [right =of h4] {};
    \node[main node] (h6) [right =of h5] {};
    \node[main node] (h7) [right =of h6] {};
    
    \node[main node] (h11) at (-0.5,-1)  {};
    \node[main node] (h12) [right =of h11] {};
    \node[main node] (h13) [right =of h12] {};
    \node[main node] (h14) [right =of h13] {};
    \node[main node] (h15) [right =of h14] {};
    \node[main node] (h16) [right =of h15] {};
    \node[main node] (h17) [right =of h16] {};
    
    \node[main node] (h21) at (-0.5,-2)  {};
    \node[main node] (h22) [right =of h21] {};
    \node[main node] (h23) [right =of h22] {};
    \node[main node] (h24) [right =of h23] {};
    \node[main node] (h25) [right =of h24] {};
    \node[main node] (h26) [right =of h25] {};
    \node[main node] (h27) [right =of h26] {};
    
    \node[main node] (h31) at (-0.5,-3)  {};
    \node[main node] (h32) [right =of h31] {};
    \node[main node] (h33) [right =of h32] {};
    \node[main node] (h34) [right =of h33] {};
    \node[main node] (h35) [right =of h34] {};
    \node[main node] (h36) [right =of h35] {};
    \node[main node] (h37) [right =of h36] {};
    
    \node[main node] (h41) at (0.25,-4)  {};
    \node[main node] (h42) [right =of h41] {};
    \node[main node] (h43) [right =of h42] {};
    \node[main node] (h44) [right =of h43] {};
    \node[main node] (h45) [right =of h44] {};
    
    \node[text width=3cm] at (-0.5, -2.5) 
    {hidden};
    
    \node[text width=3cm] at (-0.5,-5.5) 
    {visible};
    
    \node[text width=4cm] at (4.9, -5) 
    {boundary};

    \node[visible node] (v1) at (0.35, -5 ) {};
    \node[visible node] (v2) [right =of v1] {};
    \node[visible node] (v3) [right =of v2]{};
    \node[visible node] (v4) [right =of v3]{};
    \node[visible node] (v5) [right =of v4]{};
    
    \node[visible node] (v11) at (-0.4, -6 ) {};
    \node[visible node] (v12) [right =of v11] {};
    \node[visible node] (v13) [right =of v12]{};
    \node[visible node] (v14) [right =of v13]{};
    \node[visible node] (v15) [right =of v14]{};
    \node[visible node] (v16) [right =of v15]{};
    \node[visible node] (v17) [right =of v16]{}; 

    \draw (h1) -- (h11);
    \draw (h1) -- (h12);
    \draw (h1) -- (h13);
    
    \draw (h2) -- (h11);
    \draw (h2) -- (h12);
    \draw (h2) -- (h13);
    \draw (h2) -- (h14);
    
    \draw (h3) -- (h11);
    \draw (h3) -- (h12);
    \draw (h3) -- (h13);
    \draw (h3) -- (h14);
    \draw (h3) -- (h15);
    
    \draw (h4) -- (h12);
    \draw (h4) -- (h13);
    \draw (h4) -- (h14);
    \draw (h4) -- (h15);
    \draw (h4) -- (h16);
    
    \draw (h5) -- (h13);
    \draw (h5) -- (h14);
    \draw (h5) -- (h15);
    \draw (h5) -- (h16);
    \draw (h5) -- (h17);
    
    \draw (h6) -- (h14);
    \draw (h6) -- (h15);
    \draw (h6) -- (h16);
    \draw (h6) -- (h17);
    
    \draw (h7) -- (h15);
    \draw (h7) -- (h16);
    \draw (h7) -- (h17);

    \draw (h21) -- (h11);
    \draw (h21) -- (h12);
    \draw (h21) -- (h13);
    
    \draw (h22) -- (h11);
    \draw (h22) -- (h12);
    \draw (h22) -- (h13);
    \draw (h22) -- (h14);
    
    \draw (h23) -- (h11);
    \draw (h23) -- (h12);
    \draw (h23) -- (h13);
    \draw (h23) -- (h14);
    \draw (h23) -- (h15);
    
    \draw (h24) -- (h12);
    \draw (h24) -- (h13);
    \draw (h24) -- (h14);
    \draw (h24) -- (h15);
    \draw (h24) -- (h16);
    
    \draw (h25) -- (h13);
    \draw (h25) -- (h14);
    \draw (h25) -- (h15);
    \draw (h25) -- (h16);
    \draw (h25) -- (h17);
    
    \draw (h26) -- (h14);
    \draw (h26) -- (h15);
    \draw (h26) -- (h16);
    \draw (h26) -- (h17);
    
    \draw (h27) -- (h15);
    \draw (h27) -- (h16);
    \draw (h27) -- (h17);

    \draw (h21) -- (h31);
    \draw (h21) -- (h32);
    \draw (h21) -- (h33);
    
    \draw (h22) -- (h31);
    \draw (h22) -- (h32);
    \draw (h22) -- (h33);
    \draw (h22) -- (h34);
    
    \draw (h23) -- (h31);
    \draw (h23) -- (h32);
    \draw (h23) -- (h33);
    \draw (h23) -- (h34);
    \draw (h23) -- (h35);
    
    \draw (h24) -- (h32);
    \draw (h24) -- (h33);
    \draw (h24) -- (h34);
    \draw (h24) -- (h35);
    \draw (h24) -- (h36);
    
    \draw (h25) -- (h33);
    \draw (h25) -- (h34);
    \draw (h25) -- (h35);
    \draw (h25) -- (h36);
    \draw (h25) -- (h37);
    
    \draw (h26) -- (h34);
    \draw (h26) -- (h35);
    \draw (h26) -- (h36);
    \draw (h26) -- (h37);
    
    \draw (h27) -- (h35);
    \draw (h27) -- (h36);
    \draw (h27) -- (h37);

    \draw (h41) -- (h31);
    \draw (h41) -- (h32);
    \draw (h41) -- (h33);
    \draw (h41) -- (h34);
    
    \draw (h42) -- (h31);
    \draw (h42) -- (h32);
    \draw (h42) -- (h33);
    \draw (h42) -- (h34);
    \draw (h42) -- (h35);
    
    \draw (h43) -- (h32);
    \draw (h43) -- (h33);
    \draw (h43) -- (h34);
    \draw (h43) -- (h35);
    \draw (h43) -- (h36);
    
    \draw (h44) -- (h33);
    \draw (h44) -- (h34);
    \draw (h44) -- (h35);
    \draw (h44) -- (h36);
    \draw (h44) -- (h37);
    
    \draw (h45) -- (h34);
    \draw (h45) -- (h35);
    \draw (h45) -- (h36);
    \draw (h45) -- (h37);

    \draw (v1) -- (h41);
    \draw (v1) -- (h42);
    \draw (v1) -- (h43);
    
    \draw (v2) -- (h41);
    \draw (v2) -- (h42);
    \draw (v2) -- (h43);
    \draw (v2) -- (h44);
    
    \draw (v3) -- (h41);
    \draw (v3) -- (h42);
    \draw (v3) -- (h43);
    \draw (v3) -- (h44);
    \draw (v3) -- (h45);
    
    \draw (v4) -- (h42);
    \draw (v4) -- (h43);
    \draw (v4) -- (h44);
    \draw (v4) -- (h45);
    
    \draw (v5) -- (h43);
    \draw (v5) -- (h44);
    \draw (v5) -- (h45);

    \draw (v1) -- (v11);
    \draw (v1) -- (v12);
    \draw (v1) -- (v13);
    \draw (v1) -- (v14);
    
    \draw (v2) -- (v11);
    \draw (v2) -- (v12);
    \draw (v2) -- (v13);
    \draw (v2) -- (v14);
    \draw (v2) -- (v15);

    \draw (v3) -- (v12);
    \draw (v3) -- (v13);
    \draw (v3) -- (v14);
    \draw (v3) -- (v15);
    \draw (v3) -- (v16);

    \draw (v4) -- (v13);
    \draw (v4) -- (v14);
    \draw (v4) -- (v15);
    \draw (v4) -- (v16);
    \draw (v4) -- (v17);

    \draw (v5) -- (v14);
    \draw (v5) -- (v15);
    \draw (v5) -- (v16);
    \draw (v5) -- (v17);

\draw[dashed] (-0, -5.3) -- (3.4,-5.3) -- (3.4,-4.7) -- (-0,-4.7) -- (-0,-5.3);

\end{tikzpicture}\label{area_law}
\end{center}
  \caption{For an area law, the entanglement entropy scales as the number of qubits on the boundary (in the dashed rectangle). }

  \label{fig:area_law}
\end{figure}

This shows that if our quantum neural network outputs states that satisfy a volume law then asymptotically the predictions of the neural network would be no better than random guessing.  In contrast, quantum neural networks will not necessarily observe this problem if the entanglement entropy is characteristic of an area-law scaling unless the number of hidden units in the first layer becomes much larger than the number of visible units.  We, therefore, see that uncontrolled entanglement, such as that yielded by volume laws, can be catastrophic for deep quantum neural networks (i.e. for $D_h \gg D_v$) but the comparably limited entanglement yielded by area laws may be more desirable.  This means that when designing neural networks, it is vital to aim for sub-volume-law scaling. However, such states often have concise representations using matrix product states~\cite{eisert2010colloquium} and so might be no more performant than classical neural networks. Nonetheless, we show in~\ref{sec:volume-law} that such sub-volume law scalings are not typical and that almost all quantum neural networks within the ensembles we consider obey volume law scalings. 

\section{Typicality of Volume-Law Scaling}\label{sec:volume-law}
While area laws occur for certain systems, such as ground states of gapped translationally invariant Hamiltonians on lattices, we expect volume law scalings to be much more common.
This intuition can be made rigorous by making appropriate assumptions about the interactions between the visible and hidden layers in the model.  In particular, we assume that the quantum states on the joint system of the QNN approximate a Haar-random state.  In practice, this assumption is too strong as Haar random states typically require exponentially long quantum circuits to generate them.  Instead, we focus on ensembles generated by unitary $2$-designs, which model the states generated by random sequences of universal gates~\cite{harrow2009random}.

\begin{proposition}\label{prop:one}
Let $U\in \mathbb{C}^{D_vD_h\times D_vD_h}$ be drawn from a unitary $2$-design  and let $H= U^\dagger S U$ for some diagonal matrix $S\in \mathbb{C}^{D\times D}$. If either $\rho = U\ket{0}\!\bra{0}U^\dagger$ (unitary network) or $\rho = e^{-H}/{\rm Tr}(e^{-H})$ (Boltzmann machine) then any bounded operator $O_{\rm obj}\in \mathbb{C}^{D_v \otimes D_v}$ acting on the visible subspace we have that,
$$
\left|{\rm Tr}\left((O_{\rm obj}\otimes I)(\rho - I/D))\right)\right|\in O\left(\|O_{\rm obj}\|_{\infty}\sqrt{\frac{D_v}{D_h}} \right)
$$
with high probability over $U$.
\end{proposition}
\begin{proof}
Let us first examine the case of $\rho = U^\dagger \ket{0}\!\bra{0} U$.  We then have that if we take the expectation value over $U$ drawn from a $2$-design then 
\begin{align}
    &|\EX({\rm Tr}((O_{\rm obj}\otimes I )(U^\dagger \ket{0}\!\bra{0} U -I/D))| \nonumber\\
    &\le \|O_{\rm obj}\|_\infty \mathbb{E}\|{\rm Tr}_h(U^\dagger \ket{0}\!\bra{0}U -I/D )\|_1\nonumber\\
    &\le \|O_{\rm obj}\|_\infty \sqrt{D_v\mathbb{E}\|{\rm Tr}_h(U^\dagger \ket{0}\!\bra{0}U -I/D )\|_2^2}
\end{align}
Since the partial trace of a density operator is a density operator, it follows that the argument is positive definite and in turn that the result can be written as
\begin{align}
    \|O_{\rm obj}\|_\infty\sqrt{\mathbb{E}( {\rm Tr}((\rho -I/D) \otimes (\rho -I/D))(F_{vv'}\otimes I) },
\end{align}
where $F_{vv'}$ is the flip or swap operator that swaps the two visible subsystems.  

Since the result is quadratic in the probability distribution we have from the definition of a unitary $2$-design that 
\begin{align}
    &\mathbb{E}( {\rm Tr}((\rho -I/D) \otimes (\rho -I/D)) \nonumber\\
    &\qquad= \mathbb{E}_{\rm Haar}( {\rm Tr}((\rho -I/D) \otimes (\rho -I/D))
\end{align}
where $\EX_{\rm Haar}$ is the Haar expectation value.  The result then follows immediately from invoking Theorem 2 from the result of Popescu et al~\cite{popescu2006entanglement} that 
\begin{align}
    &\|O_{\rm obj}\|_\infty\sqrt{\mathbb{E}( {\rm Tr}((\rho -I/D) \otimes (\rho -I/D))(F_{vv'}\otimes I) }\nonumber\\
    &\in O\left(\|O_{\rm obj}\|_\infty \sqrt{\frac{D_v}{D_h}}\right).
\end{align}

Next let us assume that $\rho=e^{-H}/{\rm Tr}(e^{-H})$.  We have from the definition of $H$ and the previous result that 
for any eigenvector $\ket{j}$ of $H$
\begin{equation}
    \left|{\rm Tr}\left((O_{\rm obj}\otimes I)(\ket{j}\!\bra{j} - I/D))\right)\right|\in O\left(\|O_{\rm obj}\|_{\infty}\sqrt{\frac{D_v}{D_h}}\right)
\end{equation}
Since $\rho = \sum_j \ket{j}\!\bra{j} e^{-\bra{j} H \ket{j}} / {\rm Tr}(e^{-H}):= \sum_j\EX_H(\ket{j}\!\bra{j})$ the required result immediately follows by interchanging the order of the expectation values over the mixed state and over the unitary $2$-design.  These results also hold with high probability as a consequence of Markov's inequality.
\end{proof}
This shows that for both the Boltzmann machine, as well as unitary quantum networks, any observable measured on the visible layers will be indistinguishable, in expectation, to the maximally mixed state with high probability.  In other words, rather than strengthening the analogous classical model the presence of entanglement actually weakens them as the dimension of the hidden subsystem grows relative to the visible subsystem. For deep networks, we anticipate that there will be many more hidden neurons than visible neurons and hence generically entanglement is a bane not a boon for deep QNNs.

There are a number of caveats to this analysis.  First, we assume that the states in question are typical of a unitary $2$-design.  This assumption may not be appropriate if a structured ansatz is used or if the used circuits are shallow.  The next assumption is that the observable is supported on the visible system only.  The final, potential, caveat is that gradient-based optimizers may allow us to train our way out of these typical points and thereby find a way to productively leverage quantum effects.  While the first two caveats do speak to ways to escape this apparent no-go result, the ubiquity of ``entanglement induced barren plateaus'' will make the third option fail with high probability.

\section{Entanglement Induced Barren Plateaus}\label{sec:barren}
Our arguments for why gradient descent will fail to improve the quality of a training objective function due to entanglement between the visible and hidden layers follows from similar reasoning to that employed in Proposition~\ref{prop:one}.  However, the specific arguments require slightly more nuanced assumptions since we need to worry about how perturbations to the model parameters impact the resulting state.  Such assumptions are also made, for example, in the original McClean et al. work that identified Barren plateaus for unitary networks~\cite{mcclean2018barren}.  Further, while we were able to directly employ existing results from the literature of thermalization to prove Proposition~\ref{prop:one}, the necessary conditions do not hold for the gradients operator.  We state the main results below and provide an explicit proof in Appendices~\ref{app:unitary} and~\ref{app:QBM}

\subsection{Plateaus for Unitary networks}
We will first consider the case of unitary networks of the form \begin{equation}
U(\theta_1, \dots, \theta_n) := e^{-iH_n \theta_n} \dots  e^{-iH_1 \theta_1}.
\end{equation}
We consider the case where one of the parameters is shifted by a constant amount $\delta_k$ and argue about the maximum possible shift in the expectation of an observable that is supported only on the visible subsystem.  

A major challenge to analyzing what happens when shifting parameters of a unitary network is that such networks are so complicated that the impact of this perturbation is difficult to measure. An example of such an effect can be seen in the Loschmidt echo, which shows exponential sensitivity to perturbations in the parameters of complex quantum dynamics~\cite{haake1991quantum,yan2020information}.  Our solution, similar to that taken in~\cite{mcclean2018barren}, is to assume that the dynamics scrambles the states so much that almost all subsequences of the product $\prod_{j=1}^k e^{-i H_j \theta_j}$ form a unitary $2$-design.  This assumption is reasonable for a sufficiently deep random circuit~\cite{cleve2015near, haferkamp2019closing, harrow2018approximate}.  We then see that the value of the objective function is Lipshitz continuous with a constant that scales inversely with the hidden-dimension $D_h = 2^{n_h}$.  This shows that the plateau exists both for gradient descent as well as gradient-free methods~\footnote{The work of McClean et al~\cite{mcclean2018barren} can also be seen to implicitly imply barren plateaus for gradient-free methods.}. A formal statement of this intuition and the result is given below.

\begin{restatable}[Gradient in unitary networks]{theorem}{unitarygradient}
\label{thm:U_gradient}
Assume that $\rho(\theta)$ is drawn from a unitary $2$-design where $\rho(\theta)$ is generated through a unitary ansatz of the form
\begin{equation}
    \rho(\theta)=\prod_{j=1}^N e^{-iH_j \theta_j} \ket{0}\! \bra{0} \prod_{j=N}^1 e^{iH_j \theta_j}\nonumber
\end{equation}
 that acts on a Hilbert space that is the product of a hidden and visible space of dimensions $D_h$ and $D_v$ respectively.  Further, $H_k(\theta) = \prod_{j=1}^ke^{-iH_j \theta_j} H_k \prod_{j=k}^1 e^{i H_j \theta_j}$ for each $k$ obeys $\EX(H_k(\theta)\rho(\theta))=\EX(H_k(\theta))\EX(\rho(\theta))$. We then have that  
 \begin{equation}
\EX(\left|{\rm Tr}_v( O_{\rm obj} {\rm Tr}_h(\rho(\theta) )) \right|)  \nonumber
 \end{equation}
  is a Lipshitz continuous function of $\theta$ with constant $\Lambda$ obeying
  $$\Lambda \in O\left(  \|O_{\rm obj} \|_\infty \|H_k\|_\infty\sqrt{\frac{D_v}{D_h}}\right).$$
\end{restatable}

The proof of the theorem follows by using the unitary invariance of the trace norm and Hadamard's lemma to rewrite the difference between the perturbed exponential and the original exponential as a commutator series of $H_k(\theta)$ and $\rho(\theta)$.  Then by using the triangle inequality, the Cauchy-Schwarz inequality as well as the independence assumptions made above to arrive at the result.  An explicit proof is given in Appendix~\ref{app:unitary}.

\subsection{Plateaus for Boltzmann Machines}
Next, we will turn our attention to Boltzmann machines. We show that parameterized Hamiltonians drawn from a unitary ensemble also experience an entanglement induced barren plateau.  The nature of this plateau, however, differs from that of the unitary network's plateau in that the plateau occurs under reasonable assumptions if ${\rm Tr}(h_h)^2/{\rm Tr}(h_h^2) \in o(D_h)$ as we see below.

\begin{restatable}[Gradient for Boltzmann machines]{theorem}{BMgradient}
\label{thm:BM_gradient}
Assume $H\in \mathbb{C}^{D\times D}$ is a random Hermitian matrix 
drawn from an ensemble in the following manner: a diagonal matrix with eigenvalues $E_j\in\mathbb{R}$ chosen according to a probability ${\rm Pr}(E_j)$ such that $\max_k\mathbb{E}(\frac{1}{D^2}(\sum_{j\ne k}(E_j -E_k)^{-1})^2) \in {O}(\Gamma^2 )$ and then is conjugated with a unitary drawn from a distribution that is a  unitary $2$-design.  Let us then define for fixed Hermitian $H_k\in \mathbb{C}^{D\times D}$ that can be written for Hermitian $h_v,h_h$ as $H_k = h_v\otimes h_h$ and
\begin{equation}
  \rho(\theta_k):= e^{-(H+\theta_k H_k)}/{\rm Tr}(e^{-(H+\theta H_\theta)}).\nonumber  
\end{equation}
Finally, let $O_{\rm obj}\in \mathbb{C}^{D_v\times D_v}$ be a Hermitian matrix then 
\begin{equation}
\kappa:={\rm Tr}((O_{\rm obj}\otimes I_h) (\rho(\theta_k)))\nonumber
\end{equation}
 is  a differentiable function that obeys
\begin{equation}
\left|\frac{\partial \kappa}{\partial_{\theta_k}}\right|_{\theta_k=0} \!\!\!\!\! \in {O}\left( \|O_{\rm obj}\|_\infty\Gamma\|H_k\|_\infty \sqrt{\frac{D_v}{D_h}\left(\frac{{\rm Tr}(h_h)^2}{D_v{\rm Tr}(h_h^2)} + 1\right)}\right), \nonumber
\end{equation}
with high probability over the ensemble.
\end{restatable}

The proof of Theorem~\ref{thm:BM_gradient} can be found in The Appendix~\ref{app:QBM}.  The sketch of the proof is relatively simple.  We use the assumption that the eigenvectors are taken to be columns of matrices drawn from a unitary $2$-design and then use perturbation theory to argue about the perturbed $H$. The use of perturbation theory introduces the parameter $\Gamma$ that characterizes the inverse minimal gap. We then take the partial trace of the resulting perturbed eigenvectors to show that if the reduced density matrix over the hidden units of the perturbation Hamiltonian $H_k$ has zero trace then the partial trace over the hidden layers of each eigenvector remains the maximally mixed state as per Proposition~\ref{prop:one}.  This partial trace assumption is needed because if bias terms were added to the hidden units then one could disentangle them from the visible units in the ground state through the perturbation.  While such a perturbation may save the predictive power of the Boltzmann machine, it would effectively eliminate the hidden layers reverting the model to a shallow one.  With these observations, the results then follow from the use of standard inequalities and the Haar expectation value of random states given, for example, in~\cite{babbush2015chemical}. The result holds with high probability as a consequence of the Markov inequality.

In particular, we find that the gradient of the objective function with respect to terms that non-trivially act on the hidden layers are exponentially small in the number of hidden qubits since without loss of generality we may take ${\rm Tr}(h_h)=0$ for all such terms.  In contrast, the gradient with respect to the visible Hamiltonian coefficients need not be exponentially small in the number of hidden qubits.  Indeed, if we have a $k$-local random Hamiltonian where each Hamiltonian coefficient is chosen independently from a distribution that is independent of $D$ then $\Gamma \in O(\log(D)^{1-k})$ thus for any $k\ge 2$ the gradient may only be polynomially small.

A side effect of these observations is that they explain, in part, the observations in~\cite{kieferova2017tomography} that the number of hidden units included in the model did not increase the performance of Quantum Boltzmann Machines.  This can now be understood from the fact that the Gibbs states for typical Hamiltonians generate thermal states that are close to the maximally mixed state. Thus, the inclusion of hidden units typically will not be expected to increase the performance of quantum Boltzmann machines.

\section{Haar Random Unitaries}
In the previous sections, we assumed that the eigenbasis the neural networks scramble at least as effectively as a unitary $2$-design.  However, if we assume that in the case of the unitary networks the gate sequence is Haar--random or in the case of the Boltzmann machine that the basis is Haar--random, then the type of concentration that we see can be radically improved.  Specifically, Levy's lemma~\cite{popescu2006entanglement}  can be used in place of Markov's inequality to show that the vast majority of randomly selected networks will have vanishing gradients.  In particular,
\begin{lemma}[Levy]
Given a function $f:\mathbb{S}^d \mapsto \mathbb{R}$ defined on the $d$-dimensional hypersphere $\mathbb{S}^d$ and a point $\phi\in\mathbb{S}^d$ chosen uniformly at random,
$$
{\rm Prob}[|f(\phi) - \mathbb{E}(f)| \ge \epsilon] \le 2 \exp\left(\frac{-C(d+1)\epsilon^2}{\eta^2} \right),
$$
where $\eta$ is the Lipshitz constant of $f$ and $C\in \Theta(1)$.
\end{lemma}
This result ends up allowing us to use an even tighter concentration result for the systems than what is possible using Markov's inequality because it shows that a large deviation from the Haar expectation is exponentially small.  This further means that a substantial deviation from the results stated above is in fact exponentially smaller than what would be expected if we only had a $2$-design condition.  If unitary $k$-designs are used in place of $2$-designs then it should be noted that it is possible to interpolate between these two results~\cite{low2009large}, however, the bounds that arise from using this result under the assumption that we only have a $2$-design is not superior to our Markov-based analysis.

\section{Numerical Results}

We ran a series of numerical experiments summarized in Figure \ref{grad_numerics} and \ref{td_mms}  to demonstrate that our asymptotic results apply to small-sized quantum networks. 

We constructed our ansatz using the terms of a random two-local Hamiltonian model on $n$-qubits. Let $\sigma_a^j=I^{\otimes j-1}\otimes \sigma_a \otimes I^{\otimes n-j}$ for $a\in \{x,y,z\}$ and define 
\begin{equation}
   \hat{H} = \sum_i\sum_a J_a^{i} \sigma_a^i + \sum_{i<j} \sum_a \sum_b J_{a, b}^{i, j} \sigma_a^i \sigma_b^j \label{two-local_ham}
\end{equation} 
where we refer to $J_a^{i}$ as the {\it onsite} coefficients and $J_{a, b}^{i, j}$ as the {\it offsite} coefficients of our model. For the Unitary model, we exhaustively sampled from the individual terms in equations \eqref{two-local_ham} to construct the individual unitaries. For the Boltzmann model, we used $\hat{H}$ as our Hamiltonian.

 In Figure \ref{td_bound} and Figure \ref{td_mms}, we compared the trace distance scaling of the maximally mixed state and three models: the gaussian unitary ensemble model, the unitary QNN, and the Quantum Boltzmann Machine. In Figure \ref{td_bound}, we see that for an increasing number of hidden units these models will produce states close to the maximally mixed state. This result can be understood in the context of Section~\ref{sec:volume-law}.
 Figure \ref{td_mms} highlights this effect on the data histograms: as we increase the number of hidden units, we see the trace distance concentrating around zero. 
 
In Figure \ref{grad_tro}, we performed a similar analysis on the gradients of the unitary QNN. We observed an overall decrease in $\infty$-norm of the gradient vector as we increasing the size of the hidden units. We also calculated the exponential rate of decay using least square fitting. This overall decay is predicted in Theorem \ref{thm:U_gradient} as we increase the number of hidden units.

The Boltzmann Machine results are summarized in Figure \ref{grad_bm}. In order to observe gradient decay in our experiments, we need to amplify the effect of the offsite terms in relation to the onsite terms to encourage a volume-law scaling.  The emergence of these volume laws can be understood from perturbation theory since the leading order shift in an eigenvector $\ket{n}$ with eigenvalue $E_n$ is proportional to $\sum_{j\ne n} \ket{j}\!\bra{j}H_k \ket{n}/(E_n -E_j)$.  This shows that if we take $\ket{n}$ to be an eigenstate of the $1$-body terms in the Hamiltonian then the entanglement generated by $H_k$ is suppressed by the energy gaps between these states.  We, therefore, choose these magnitudes to be small so that significant entanglement can be introduced in the eigenstates despite the small values of $D$ that can be explored on a classical computer. This phenomenon is predicted in Theorem~\ref{thm:BM_gradient}.

\begin{figure*}[htb]
  \centering
   \subfloat[
    ]{
    \includegraphics[width=0.33\textwidth]{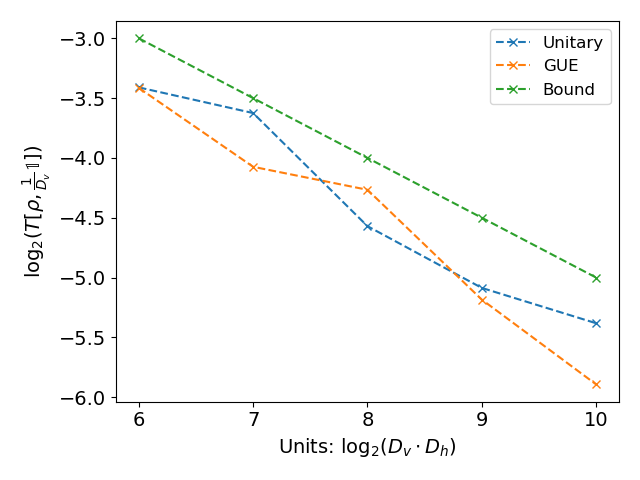}\label{td_bound}
  }
  \subfloat[]{
    \includegraphics[width=0.33\textwidth]{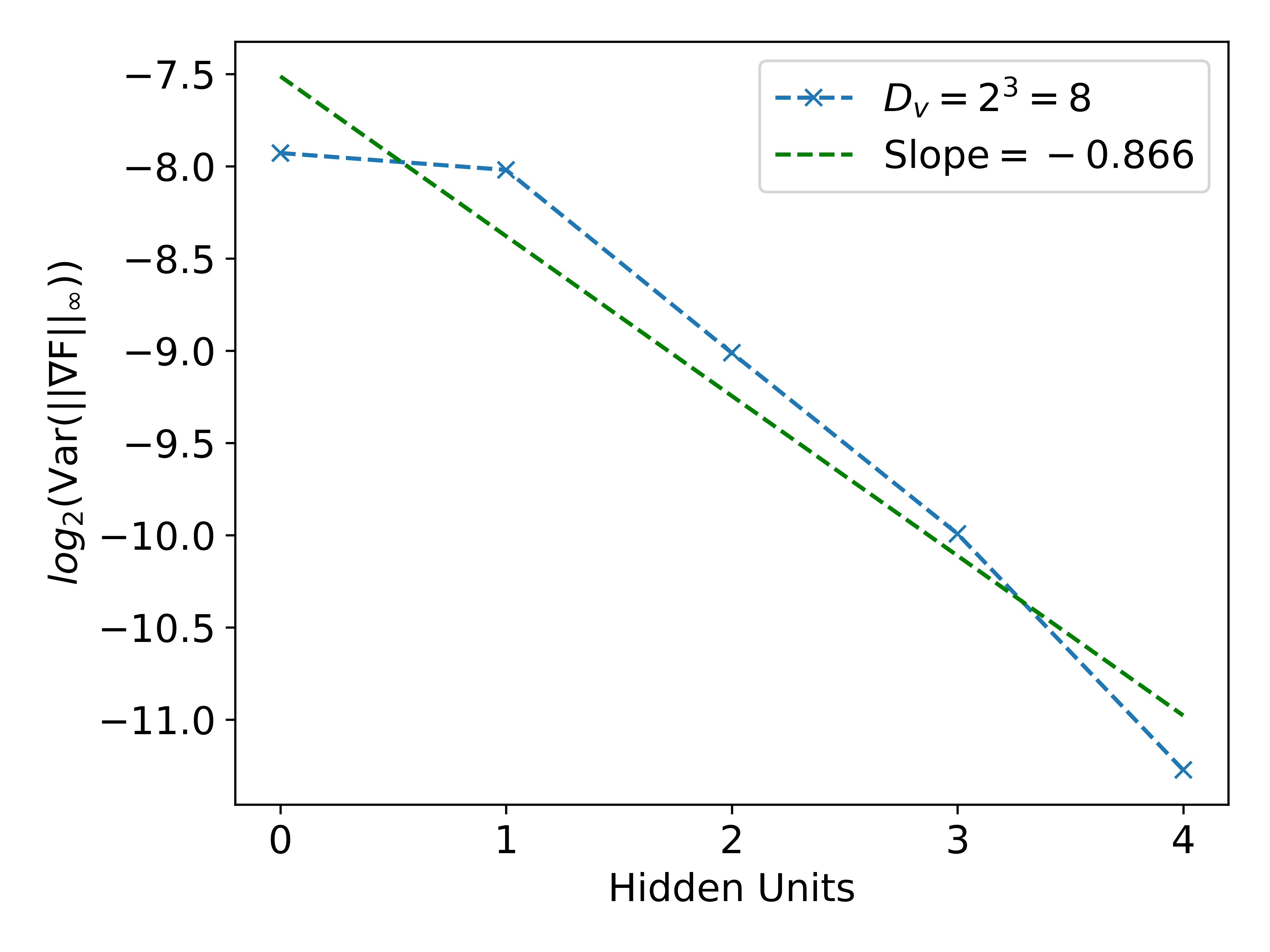}\label{grad_tro}
  }
  \subfloat[]{
   \includegraphics[width=0.33\textwidth]{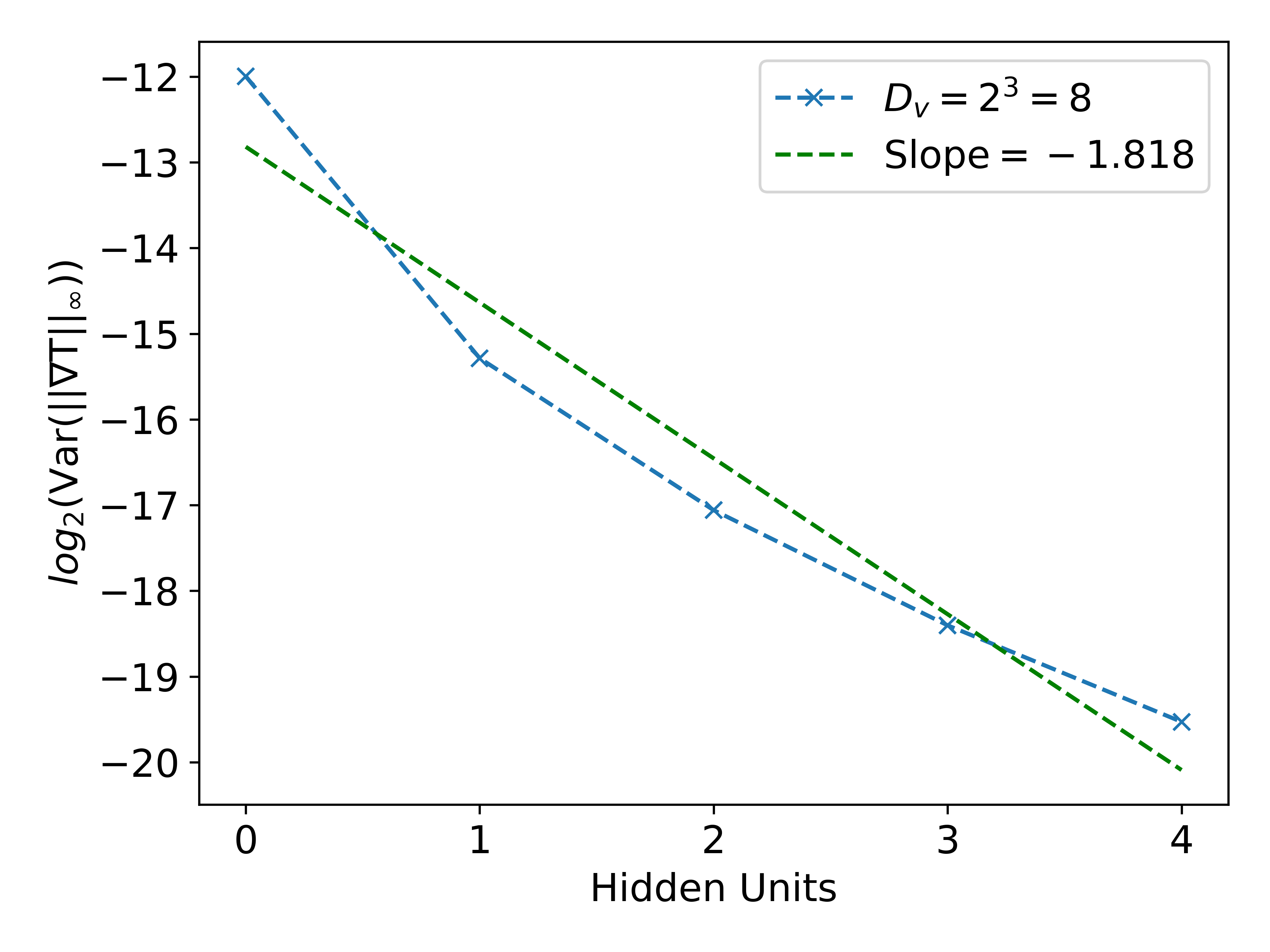}\label{grad_bm}
   }
  \caption[]{
  (a) Log-Log plot showing the of trace distance data in relation to the bound. The blue and orange marked values correspond to the estimated maximum peak of the data histograms where $D_v=2^1=2$. The green marked values correspond to the bounds we obtain after substituting in \(\mathbb{E}[T(\rho, I/D)]\leq 1/2\sqrt{D_v/D_h}\).   (b-c) Semi-log plot highlighting the decay in the variance of the $\infty$-norm of the gradient vector over an ensemble of initialized models. The dashed blue represents the average of $1000$ model instances. The dash green line represents is the best fit obtained from least squares.  (b) Gradient estimates for the Unitary Model. (c) Gradient estimates for the normalized Quantum Boltzmann Machine.
  }
  \label{grad_numerics}
\end{figure*}

In Figure \ref{grad_tro}, we generated a fixed thermal state $e^{-\hat{H}}/Z$ with onsite coefficients drawn from a normal distribution with mean 0, variance 0.01, i.e. $\mathcal{N}(0,0.01)$ and $\mathcal{N}(0,1)$ for the offsite coefficients. We then proceeded to estimate the gradient vector of the Fidelity, $F$, between our model and this fixed state using finite differences. We generated $1000$ instances of our Unitary model by initializing all the coefficients with samples from $\mathcal{N}(0,1)$. The figure shows a decrease in the variance $\infty$-norm of the gradient vector on a semilog scale.

In Figure \ref{grad_bm}, we estimated the gradient vector of the trace distance, $T$, between our model and its perturbation for each parameter using finite differences. The onsite coefficients where drawn from $\mathcal{N}(0,0.01)$ and the offsite coefficients from $\mathcal{N}(0,1)$. Moreover, normalized the Hamiltonian by its operator norm.

\begin{figure*}[!htb]
  \centering%
  \subfloat[Real-Time Gaussian Unitary ensemble]{%
    \includegraphics[width=0.33\textwidth]{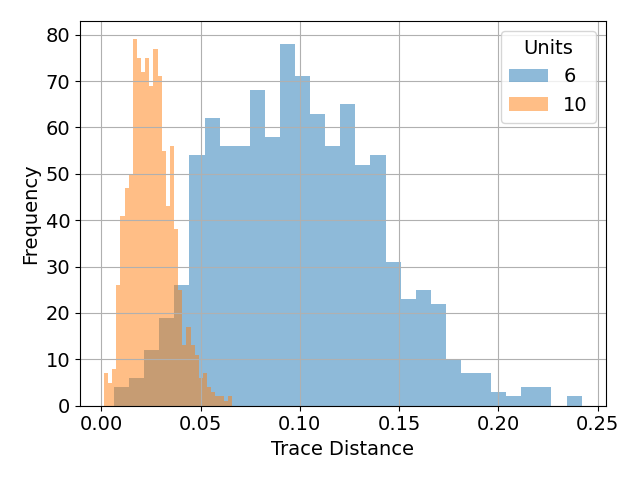}%
  }
  \subfloat[
    Unitary Model
    ]{\includegraphics[width=0.33\textwidth]{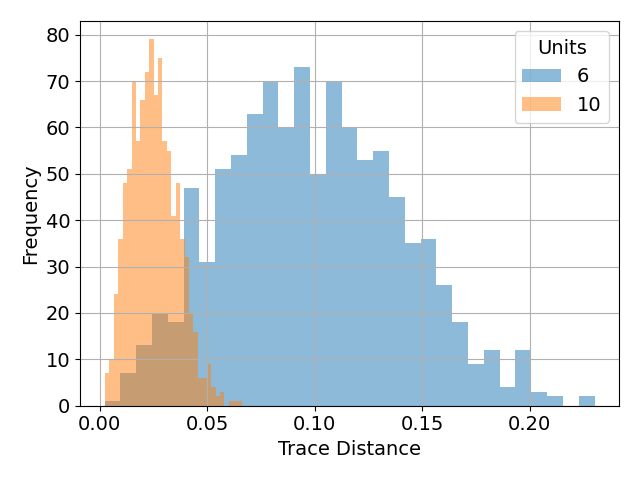}%
  }
   \subfloat[Normalized Boltzmann Machine]{%
    \includegraphics[width=0.33\textwidth]{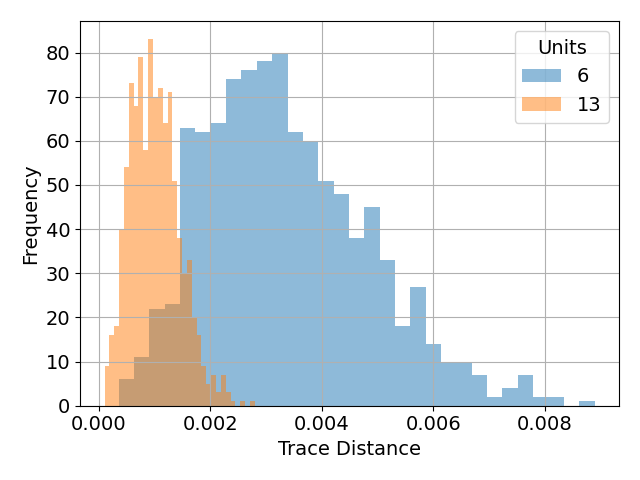}%
  }
 
  \caption[]{Computed the trace distance between the reduce density of our models and the maximally mixed state for 1000 instances. The models considered have only one visible unit i.e. $D_v=2^1=2$.  (a) Empirical trace distance distribution of a real-time evolution ($t=10$) of Hamiltonians drawn from the Gaussian Unitary Ensemble (GUE). (b) Empirical trace distance distribution of the unitary model. All coefficients are drawn from a uniform distribution over $[0, 1)$. (c) Empirical trace distance distribution of the quantum Boltzmann machine. The on-set coefficients, $J_a^{i}$,  are drawn $\mathcal{N}(0,0.01)$. The off-set coefficients, $J_{a, b}^{i, j}$, are drawn from $\mathcal{N}(0,1)$.  Moreover, the Hamiltonian is normalized by its operator norm. }
  \label{td_mms}
\end{figure*}

\section{Conclusion}
We showed that for Haar-random pure states and thermal states of random Hamiltonians, the gradient of an observable objective function will be vanishing exponentially with the number of hidden units. This shows that common types of QNNs are not only generically difficult to train via local optimization methods but also that adding hidden units will not always increase the power of QNNs.  Indeed, asymptotically we see that unless the states generated satisfy an area law such hidden neurons will likely be harmful.

One can prevent these entanglement induced barren plateaus by violating any of the assumptions in our proofs. The first is to choose an atypical initial state which has been already explored in~\cite{grant2019initialization}. Next, one could try to depart from the use of gradient-based optimization to train such quantum models. However, it is unlikely that without knowledge of the global properties of the training objective function that such methods would succeed in light of Proposition~\ref{prop:one}. Lastly, one can train models using an objective function that does not correspond to an observable and is independent of the density operator.

Of the three approaches, it is this last approach that we advocate greater attention be paid to in quantum machine learning.  One tactic that can be used to circumvent our pessimistic results is to begin a discriminative learning task by first training generatively according to a quantity such as the quantum relative entropy~\cite{kieferova2017tomography,wiebe2019generative} which is non-linear in the quantum state $\rho$.  We will show in subsequent work that this quantum generative pre-training approach can be used to successfully train both Boltzmann machines and unitary networks and thereby mitigate some of the challenges identified here for training deep QNNs.

As a final point, it is important to recognize that while entanglement is a powerful tool to add to our models, it must be used like a scalpel and not a sledgehammer.  Quantum properties such as entanglement may be harmful if not surgically deployed and judicially used.  Understanding the role that such quantum effects have on a model is very likely necessary~\cite{wiebe2020key} if we are to build quantum models that can successfully leverage quantum effects.

\acknowledgments
We thank Michael Bremner, Jarrod McClean and Alessandro Rogero for helpful discussions and feedback.
M\'aria Kieferov\'a acknowledges funding from ARC Centre of Excellence for Quantum Computation and Communication Technology
(CQC2T), project number CE170100012. Support for C. Ortiz Marrero and Nathan Wiebe for the numerical studies was provided by the Laboratory Directed Research and Development Program at Pacific Northwest National Laboratory, a multi-program national laboratory operated by Battelle for the U.S. Department of Energy, Release No. PNNL-SA-157287 and the theoretical work on this project by NW was supported by the U.S. Department of Energy, Office of Science, National Quantum Information Science Research Centers, Co-Design Center for Quantum Advantage under contract number DE-SC0012704. Additional logistical support for Nathan Wiebe was provided by the Google Research Award.

\bibliographystyle{unsrt}  
\bibliography{references}  

\onecolumngrid

\appendix


\section{Proof of Unitary Network Gradient}\label{app:unitary}
Here we provided a complete proof of Theorem \ref{thm:U_gradient}.

\unitarygradient*

\begin{proof}
First using the definition that \begin{align}
    \rho(\theta+\delta_k) :=& \left(\prod_{j=1}^k e^{-iH_j\theta_j}e^{-iH_k\delta_k} \prod_{j>k}e^{-i H_j \theta_j}\right) \rho_0 
    \times\left(\prod_{j=1}^k e^{-iH_j\theta_j}e^{-iH_k\delta_k} \prod_{j>k}e^{-i H_j \theta_j}\right)^\dagger
\end{align}

we then wish to analyze the distribution over $\theta$ of 
$
    \left|{\rm Tr}_v( O_{\rm obj} {\rm Tr}_h(\rho(\theta) - \rho(\theta+\delta_k))) \right|,
$
under the assumption that the unitaries satisfy a $2$-design condition.

Using Hadamard's Lemma we can express the distance between the difference between the expectation values is

\begin{align}
    &{\rm Tr}_v\left(O_{\rm obj}{\rm Tr}_h\left( \rho(\theta) - \left(\prod_{j=1}^ke^{-iH_j \theta_j} e^{-iH_k |\delta_k|} \prod_{j=k}^1 e^{i H_j \theta_j}\right) \rho(\theta)\left(\prod_{j=1}^ke^{-iH_j \theta_j} e^{-iH_k |\delta_k|} \prod_{j=k}^1 e^{i H_j \theta_j}\right)^\dagger \right) \right)\nonumber\\
    &\qquad:=\left({\rm Tr}_v\left( O_{\rm obj}{\rm Tr}_h\left( \rho(\theta) - e^{-iH_k(\theta) |\delta_k|}\rho(\theta) e^{iH_k(\theta) |\delta_k|}  \right) \right)\right) ={\rm Tr}_v\left( O_{\rm obj}\sum_{q=1}^\infty \frac{ {\rm Tr}_h({\rm Ad}_{-iH_k(\theta)}^q(\rho(\theta)))|\delta_k|^q }{q!}  \right)  \label{eq:networkError}
\end{align}

Next, making the assumption that $H_k(\theta)$ and $\rho(\theta)$ are uncorrelated we can further simplify this result. Our exposition will now follow that of Popescu, Short and Winter~\cite{popescu2006entanglement}; which we modify to deal with to show that a concentration of measure exists for the commutators of $H_k(\theta)$ and $\rho(\theta)$.

We will now work under the assumption that the expectation values are independent.  We further denote the expectation value over the Hamiltonian as $\EX_H$ and the expectation value over the state as $\EX_\phi$. If this independence assumption holds then we need to argue about the magnitude of terms of the form
$\EX_\phi({\rm Tr_v}({\rm Tr}_h (H_{k}(\theta)\rho(\theta))^2)).$
We can estimate this by introducing two copies of the quantum state and linking both terms through the use of a flip operator $F_{vv'}$ such that
\begin{equation}
    F_{vv'}= \sum_{vv'} \ket{v'}\!\bra{v} \otimes \ket{v}\!\bra{v'}.
\end{equation}
In the following, we will use this notation primed indices to refer to the visible and hidden subsystems of the first and second copies respectively.

The commutators in general consist of many different products of $H_k$ and the state operator.  Below we argue about their form in generality.  Let us assume that $p_1,p_2,q_1,q_2$ are positive integers.  We then wish to compute the product of traces of of the form ${\rm Tr}_h(H_k(\theta)^{p_1} \rho H_k(\theta)^{p_2}{\rm Tr}_h(H_k(\theta)^{q_1} \rho H_k(\theta)^{q_2})$.
By applying the flip operator and taking the quantum state $\rho(\theta)$ to be $\ket{\phi}\!\bra{\phi}$

\begin{align}
&\EX_\phi({\rm Tr}_h(H_k(\theta)^{p_1}) \rho H_k(\theta)^{p_2}{\rm Tr}_h(H_k(\theta)^{q_1} \rho H_k(\theta)^{q_2}))\nonumber\\
&={\rm Tr_{vv'}}(\EX_\phi(({\rm Tr}_h (H_{k}^{p_1}(\theta)\rho(\theta)H_{k}^{p_2}(\theta))\otimes {\rm Tr}_h (H_{k}^{q_1}(\theta)\rho(\theta)H_{k}^{q_2}(\theta)))F_{vv'} ))\nonumber\\
&={\rm Tr}_{vv'}{\rm Tr}_{hh'}(\EX_\phi(({\rm Tr}_h (H_{k}^{p_1}(\theta)\rho(\theta)H_{k}^{p_2}(\theta))\otimes {\rm Tr}_h (H_{k}^{q_1}(\theta)\rho(\theta)H_{k}^{q_2}(\theta)))(F_{vv'}\otimes I )))\nonumber\\
&={\rm Tr}(\EX_\phi( (\ket{\phi}\!\bra{\phi}\otimes  \ket{\phi}\!\bra{\phi})(H_{k}^{p_1}(\theta)\otimes H_{k}^{q_1}(\theta))(F_{vv'} \otimes I)(H_{k}^{p_2}(\theta)\otimes H_{k}^{q_2}(\theta)) )).
\end{align}

The next step in this is to recognize that the above tensor products if $\ket{\phi}\!\bra{\phi}$ are a symmetric quantum state.  Therefore if we express the state as the sum of its anti-symmetric component and its symmetric component then the anti-symmetric component must be zero~\cite{popescu2006entanglement}.  We then see from the fact that $\rho(\theta)$ is assumed to be drawn from a unitary $2$-design that the expectation value is unitarily invariant and we can then follow the arguments laid out in~\cite{popescu2006entanglement} that
\begin{align}
    &\EX_\phi({\rm Tr}_h(H_k(\theta)^{p_1}) \rho H_k(\theta)^{p_2}{\rm Tr}_h(H_k(\theta)^{q_1} \rho H_k(\theta)^{q_2})) \nonumber\\
    &= \frac{2D^2}{D(D+1)} {\rm Tr}\Biggr(\EX_\phi\Biggr( \left(\frac{\Pi^{\rm sym}}{D^2}\right)(H_{k}^{p_1}(\theta)\otimes H_{k}^{q_1}(\theta))\nonumber\\
    &\qquad\times (F_{vv'} \otimes I)(H_{k}^{p_2}(\theta)\otimes H_{k}^{q_2}(\theta)) \Biggr)\Biggr) 
\end{align}
Next, if we define the flip operator on the dilated space including the hidden and visible units to be $F_{rr'}= F_{vv'} \otimes F_{hh'}$ then we can express $\Pi^{\rm sym} = \frac{1}{2}(I +F_{rr'})$.  Finally using the properties of the flip operator we find that we can  write 

\begin{align}
    &{\rm Tr}_v(\EX_\phi({\rm Tr}_h(H_k(\theta)^{p_1}) \rho H_k(\theta)^{p_2}{\rm Tr}_h(H_k(\theta)^{q_1} \rho H_k(\theta)^{q_2})))\nonumber\\
    &=\frac{D^2}{D(D+1)}{\rm Tr}\left(\frac{(H_k^{p_1}(\theta) \otimes H_k^{q_1}(\theta))(F_{vv'}\otimes I)(H_k^{p_2}(\theta) \otimes H_k^{q_2}(\theta))}{D^2} \right)\nonumber\\
    &\qquad+\frac{D^2}{D(D+1)}{\rm Tr}\left(\frac{(F_{vv'}\otimes F_{hh'})(H_k^{p_1}(\theta) \otimes H_k^{q_1}(\theta))(F_{vv'}\otimes I)(H_k^{p_2}(\theta) \otimes H_k^{q_2}(\theta))}{D^2} \right)\nonumber\\
    &=\frac{D^2}{D(D+1)}{\rm Tr}\left(\frac{(H_k^{p_1}(\theta) \otimes H_k^{q_1}(\theta))(F_{vv'}\otimes I)(H_k^{p_2}(\theta) \otimes H_k^{q_2}(\theta))}{D^2} \right)\nonumber\\
    &\qquad+\frac{D^2}{D(D+1)}{\rm Tr}\left(\frac{(H_k^{q_1}(\theta) \otimes H_k^{p_1}(\theta))(I\otimes F_{hh'})(H_k^{p_2}(\theta) \otimes H_k^{q_2}(\theta))}{D^2} \right)\nonumber\\
    &=\frac{D^2}{D(D+1)}{\rm Tr}\left(\frac{(F_{vv'}\otimes I)(H_k^{p_1+p_2}(\theta) \otimes H_k^{q_1+q_2}(\theta))}{D^2} \right)+\frac{D^2}{D(D+1)}{\rm Tr}\left(\frac{(I\otimes F_{hh'})(H_k^{q_1+p_2}(\theta) \otimes H_k^{p_1+q_2}(\theta))}{D^2} \right)\nonumber\\
    &=\frac{D^2}{D(D+1)}{\rm Tr}_v\left(\frac{{\rm Tr}_h(H_k^{p_1+p_2}(\theta)) {\rm Tr}_h( H_k^{q_1+q_2}(\theta))}{D^2} \right)+\frac{D^2}{D(D+1)}{\rm Tr}_h\left(\frac{{\rm Tr}_v(H_k^{q_1+p_2}(\theta)) {\rm Tr}_v( H_k^{p_1+q_2}(\theta))}{D^2} \right)\nonumber\\
\end{align}

We, therefore, have from the triangle inequality that

\begin{align}
    &\left| {\rm Tr}_v(\EX_\phi({\rm Tr}_h(H_k(\theta)^{p_1}) \rho H_k(\theta)^{p_2}{\rm Tr}_h(H_k(\theta)^{q_1} \rho H_k(\theta)^{q_2}))) - \frac{D^2}{D(D+1)}{\rm Tr}_v\left(\frac{{\rm Tr}_h(H_k^{p_1+p_2}(\theta)) {\rm Tr}_h( H_k^{q_1+q_2}(\theta))}{D^2} \right)\right|\nonumber\\
    &\qquad\le \left|{\rm Tr}_h\left(\frac{{\rm Tr}_v(H_k^{q_1+p_2}(\theta)) {\rm Tr}_v( H_k^{p_1+q_2}(\theta))}{D^2} \right)\right| \le \frac{{\rm Tr}(|H_k^{q_1+q_2}(\theta)|)\|{\rm Tr}_v (H_k(\theta)^{q_1+q_2})\|_\infty}{D^2}\le \frac{\|H_k\|_\infty^{p_1+p_2+q_1+q_2}}{D_h}.\label{eq:prodbd}
\end{align}

Here the last inequality follows from the fact that the Schatten infinity-norm is unitarily invariant and thus $\|H_k(\theta)\|_\infty = \|H_k\|_\infty$.

 Next let us consider the expectation value  for one of the terms in the expansion
\begin{align}
    &\EX(\|{\rm Tr}_h ({\rm Ad}_{-i H_k(\theta)}^q(\rho(\theta)))\|_1) \nonumber\\
    &\le \EX( \sqrt{D_v} \|{\rm Tr}_h ({\rm Ad}_{-i H_k(\theta)}^q(\rho(\theta)))\|_2) \nonumber\\
    &\le \sqrt{D_v\EX(\|{\rm Tr}_h ({\rm Ad}_{H_k(\theta)}^q(\rho(\theta)))\|_2^2)}
\end{align}

Every term in ${\rm Ad}_{H_k(\theta)}^q(\rho(\theta))$ consists of $q$ $H_k(\theta)$ and further $2^{q-1}$ terms have positive coefficient and $2^{q-1}$ terms have negative coefficient.  The proof of this fact is inductive.  For $q=1$, \begin{equation}
    {\rm Ad}_{H_k(\theta)}(\rho(\theta)) = H_k(\theta) \rho(\theta) - \rho(\theta) H_k(\theta),
\end{equation}
which demonstrates the base case of $q=1$. Now assume that the claim is valid for $q=p$ we then have that
\begin{equation}
    {\rm Ad}_{H_k(\theta)}^{p+1}(\rho(\theta)) = H_k(\theta){\rm Ad}_{H_k(\theta)}^{p}(\rho(\theta)) -{\rm Ad}_{H_k(\theta)}^{p}(\rho(\theta))H_k(\theta)\label{eq:stuff}
\end{equation}
The induction step immediately follows from this observation and it is clear that the claim is valid for all $q$.

Now if we expand $({\rm Tr}_h({\rm Ad}_{H_k(\theta)}^{q} (\rho(\theta))))^2$ using the linearity of the partial-trace operation we find that each term is of the form ${\rm Tr_h}(H_k^{p_1}(\theta) \rho(\theta) H_k^{p_2}(\theta)){\rm Tr_h}(H_k^{q_1}(\theta) \rho(\theta) H_k^{q_2}(\theta))$ where $p_1+p_2=q=q_1+q_2$.  The expression in~\eqref{eq:prodbd} then shows us that we can replace each term with $\frac{D^2}{D(D+1)} {\rm Tr}_v \left(\frac{({\rm Tr}_h(H_k^q(\theta)))^2}{D^2} \right)$ while incurring a small error.  Importantly, this value is independent of $p_1,p_2,q_1,q_2$.  Thus since there are $2^{q-1}$ such terms with negative coefficient and $2^{q-1}$ with positive coefficient for each of the partial traces there are similarly $2^{2q-1}$ terms with negative coefficient in the expansion and $2^{2q-1}$ with positive coefficient.  Ergo, the sums over all such terms present in the adjoint is zero up to the small error terms given in~\eqref{eq:prodbd}.  Thus we have that,
\begin{equation}
    \sqrt{D_v\EX(\|{\rm Tr}_h ({\rm Ad}_{H_k(\theta)}^q(\rho(\theta))\|_2^2)} \le (2\|H_k\|_\infty)^{q}\sqrt{\frac{D_v}{D_h}}\label{eq:unitaryArg}
\end{equation}
Next from~\eqref{eq:stuff} we have that
\begin{align}
    &\EX\left|{\rm Tr}_v\left( O_{\rm obj}\sum_{q=1}^\infty \frac{ {\rm Tr}_h({\rm Ad}_{-iH_k(\theta)}^q(\rho(\theta)))|\delta_k|^q }{q!}  \right)\right|  \nonumber\\
    &\le \|O_{\rm Obj}\|_\infty \sum_{q=1}^\infty \EX\left(\left\|  \frac{ {\rm Tr}_h({\rm Ad}_{-iH_k(\theta)}^q(\rho(\theta)))|\delta_k|^q }{q!}\right\|_1 \right)\nonumber\\
    &\le \|O_{\rm Obj}\|_\infty \sum_{q=1}^\infty \sqrt{D_v\EX\left(\left\|  \frac{ {\rm Tr}_h({\rm Ad}_{-iH_k(\theta)}^q(\rho(\theta)))|\delta_k|^q }{q!}\right\|_2^2 \right)}\nonumber\\
    &\le |\delta_k| \|O_{\rm obj}\|_\infty\|H_k\|_\infty e^{2\|H_k\|_\infty |\delta_k|}\sqrt{\frac{D_v}{D_h}},\nonumber\\
    &\in O\left(|\delta_k| \|O_{\rm Obj}\|_\infty \|H_k\|_\infty \sqrt{\frac{D_v}{D_h}} \right)
\end{align}
where we have used the assumption that $\|H_k\|_\infty |\delta_k|\in O(1)$.  From this our claim about the Lipshitz constant immediately follows from the definition of Lipshitz continuity and from~\eqref{eq:networkError}.
\end{proof}

\section{Proof of Quantum Boltzmann Machine Gradient}\label{app:QBM}
Here we provide a complete proof for the gradient of a quantum Boltzmann machine. We can always assume that the Hamiltonian is traceless. Indeed, for any Hamiltonian $H'$ with a non-zero trace, we can introduce a Hamiltonian $H=H'-\alpha \mathbb{1}$ such that $\text{Tr}(H)=0$ and $H$ leads to the same thermal state
\begin{equation}
    \rho_{thermal}=\frac{e^{-Ht}}{\text{Tr}(e^{-Ht})} = \frac{e^{-H't + \alpha \mathbb{1} t}}{\text{Tr}(e^{-H't + \alpha \mathbb{1} t} )} = \frac{ e^{\alpha t} e^{-H't  }}{e^{\alpha  t}\text{Tr}(e^{-H't} )} = \frac{e^{-H't}}{\text{Tr}(e^{-H't})} .
\end{equation}

\BMgradient*

\begin{proof}

First note that if we begin with an observable $O_{\rm obj}$ acting on the visible subspace then the difference between the observable for $\rho(\theta) = e^{-(H+\theta_k H_k)}/Z(\theta_k)$ and $\rho(0):=\rho$
\begin{align}
    \mathbb{E}(|{\rm Tr}( (O_{\rm obj}\otimes I_h)\rho(\theta_k)) - {\rm Tr}( (O_{\rm obj}\otimes I_h)\rho)|)&\le \|O_{\rm Obj}\|_\infty \EX( \|{\rm Tr}_h(\rho(\theta_k) - \rho)\|_1)\nonumber\\
    &\le\|O_{\rm Obj}\|_\infty \sqrt{D_v\EX( \|{\rm Tr}_h(\rho(\theta_k) - \rho)\|_2^2)}.\label{eq:diff}
\end{align}
Therefore just like the unitary network case, we will now focus our attention to bounding the expectation value of the difference between the density operators.  The main difference here is that the density operators are defined via imaginary time-evolution rather than real time.

From Taylor's theorem, we have that if the Hamiltonian $H+sH_k$ has no level crossings on the interval $s\in [0,\theta_k]$ then to order $O(\theta_k^2)$ the eigenvectors of $H+\theta_k H_k$ can be identified using perturbation theory.  In particular, for any $p\in \{0,\ldots,D-1\}$ let $\ket{p}$ be an eigenvector of $H$ with eigenvalue $E_p$ then the  eigenvector $\ket{n'}$ of $H+\theta_k$ that corresponds to the eigenvector $\ket{n}$ of $H$ can be expressed as
\begin{equation}
    \ket{n'} = \ket{n} + \theta_k\sum_{j\ne n}\frac{ \ket{j}\bra{j}H_k \ket{n}}{E_n-E_j}+O(\theta_k^2)
\end{equation}
This implies that
\begin{equation}
    {\rm Tr}_h(\ket{n'}\!\bra{n'} - \ket{n}\!\bra{n}) ={\rm Tr}_h\left(\theta_k\sum_{j\ne n}\frac{ \ket{j}\bra{j}H_k \ket{n}\!\bra{n}}{E_n-E_j}+\frac{ \ket{n}\!\bra{n}H_k \ket{j}\!\bra{j}}{E_n-E_j}\right)+O(\theta_k^2) 
\end{equation}
Next let us write for any $\ell \in \{0,\ldots,D-1\}$ the eigenvector $\ket{\ell} = \sum_{pq} \alpha_{pq}^{(\ell)} \ket{pq}$,
where $\alpha_{pq}^{(\ell)}$ is a complex number and $\ket{pq} := \ket{p}_v \otimes \ket{q}_h$ for some appropriate basis for the visible and hidden subsystems.  The expectation value over the state vectors can then be thought of as an average of these coefficients.

There are many choices that can be made for the eigenbasis that we further exploit the fact that $H_k := h_v \otimes h_h$ to choose the bases of the visible and hidden subsystems to diagonialze $h_v$ and $h_h$.  Thus we can state $h_v \otimes h_h \ket{pq} := \lambda_{pq} \ket{pq}$ for $\lambda_{pq} \in \mathbb{R}$.

With these choices in place we can write
\begin{align}
    {\rm Tr}_h\left(\sum_{j\ne n}\frac{ \ket{j}\bra{j}H_k \ket{n}\!\bra{n}}{E_n-E_j} \right)&= \sum_{j\ne n}{\rm Tr}_h\left(\sum_{pqrstuvw} \frac{\alpha_{pq}^{(j)} \alpha_{rs}^{*(j)}\alpha_{tu}^{(n)}\alpha_{vw}^{*(n)}}{E_n-E_j}  \ket{pq}\!\bra{rs} H_k\ket{tu}\!\bra{vw} \right)\nonumber\\
    &= \sum_{j\ne n}{\rm Tr}_h\left(\sum_{pqrsvw} \frac{\alpha_{pq}^{(j)} \alpha_{rs}^{*(j)}\alpha_{rs}^{(n)}\alpha_{vw}^{*(n)}\lambda_{rs}}{E_n-E_j}  \ket{pq}\!\bra{vw} \right)\nonumber\\
    &= \sum_{j\ne n}\left(\sum_{pqrsv} \frac{\alpha_{pq}^{(j)} \alpha_{rs}^{*(j)}\alpha_{rs}^{(n)}\alpha_{vq}^{*(n)}\lambda_{rs}}{E_n-E_j}  \ket{p}\!\bra{v} \right)
\end{align}

\begin{align}
    &{\rm Tr}_h(\ket{n'}\!\bra{n'} - \ket{n}\!\bra{n})^2= \left(\sum_{j\ne n}\left(\sum_{pqrsv} \frac{(\alpha_{pq}^{(j)} \alpha_{rs}^{*(j)}\alpha_{rs}^{(n)}\alpha_{vq}^{*(n)}+\alpha_{pq}^{(n)} \alpha_{rs}^{*(n)}\alpha_{rs}^{(j)}\alpha_{vq}^{*(j)})\lambda_{rs}}{E_n-E_j}  \ket{p}\!\bra{v} \right) \right)^2+O(\theta_k^4)\nonumber\\
    &=\sum_{j,j'\ne n}\sum_{\stackrel{pqrsv}{q'r's'v'}} \frac{(\alpha_{pq}^{(j)} \alpha_{rs}^{*(j)}\alpha_{rs}^{(n)}\alpha_{vq}^{*(n)}+\alpha_{pq}^{(n)} \alpha_{rs}^{*(n)}\alpha_{rs}^{(j)}\alpha_{vq}^{*(j)})(\alpha_{vq'}^{(j)} \alpha_{r's'}^{*(j)}\alpha_{r's'}^{(n)}\alpha_{v'q'}^{*(n)}+\alpha_{vq'}^{(n)} \alpha_{r's'}^{*(n)}\alpha_{r's'}^{(j)}\alpha_{v'q'}^{*(j)})\lambda_{rs}\lambda_{r's'}}{(E_n-E_j)(E_n-E_{j'})}  \ket{p}\!\bra{v'} +O(\theta_k^4)
\end{align}
There are a total of $8$ terms that arise when we expand the above products.  Let us consider the first case which emerges in the expectation value of the trace of the previous result.  Here we will invoke the fact that the eigenvectors are sampled from a unitary $2$-design, which means that any quantity that is at most quadratic in the probability will have an expectation value that coincides with the Haar average.  A final point to note, is that as a consequence of unitary invariance and the discussion contained in \cite[Appendix A]{babbush2015chemical}, the expectation value of the product of any two terms is zero unless all of their indices match.  Further, up to relative errors that are $O(1/D)$, the expectation values of the coefficients are independent of each other.  This implies that if ${\rm Tr}(h_h)=0$ then
\begin{align}
    &\EX\left(\sum_{j,j'\ne n}\sum_{\stackrel{pqrsv}{q'r's'v'}} \frac{\alpha_{pq}^{(j)} \alpha_{rs}^{*(j)}\alpha_{rs}^{(n)}\alpha_{vq}^{*(n)}\alpha_{vq'}^{(j)} \alpha_{r's'}^{*(j)}\alpha_{r's'}^{(n)}\alpha_{pq'}^{*(n)}\lambda_{rs}\lambda_{r's'}}{(E_n-E_j)(E_n-E_{j'})} \right)\nonumber\\
    &=O\Biggr(\Biggr(\sum_{j,j'\ne n}\sum_{pqq'} \frac{\EX(|\alpha_{pq}^{(j)}|^2)\EX(|\alpha_{pq}^{(n)}|^2)\EX(|\alpha_{pq'}^{(j')}|^2)\EX(|\alpha_{pq'}^{(n)}|^2)\lambda_{pq}\lambda_{pq'}}{(E_n-E_j)(E_n-E_{j'})} \Biggr)\nonumber\\
    &\qquad+\Biggr(\sum_{j,j'\ne n}\sum_{pqv} \frac{\EX(|\alpha_{pq}^{(j)}|^2)\EX(|\alpha_{pq}^{(n)}|^2)\EX(|\alpha_{vq}^{(n)}|^2)\EX(|\alpha_{pq}^{(j')}|^2)\lambda_{pq}^2}{(E_n-E_j)(E_n-E_{j'})} \Biggr)\Biggr)\nonumber\\
    &\qquad+\Biggr(\sum_{j,j'\ne n}\sum_{pq} \frac{\EX(|\alpha_{pq}^{(n)}|^4)\EX(|\alpha_{pq}^{(j)}|^2)\EX(|\alpha_{pq}^{(j')}|^2)\lambda_{pq}^2}{(E_n-E_j)(E_n-E_{j'})} \Biggr)\Biggr)\nonumber\\
    &=O\Biggr(\frac{1}{D^4} \Biggr(\sum_{j\ne n} \frac{1}{E_n - E_j} \Biggr)^2\Biggr(({\rm Tr}(h_v^2) ({\rm Tr}(h_h))^2 + (D_v+1){\rm Tr}(H_k^2) \Biggr)\nonumber\\
    &=O\Biggr(\frac{1}{D^4} \Biggr(\sum_{j\ne n} \frac{1}{E_n - E_j} \Biggr)^2\Biggr(\left(\frac{{\rm Tr}(h_h)^2}{{\rm Tr}(h_h^2)} + D_v\right){\rm Tr}(H_k^2) \Biggr)\nonumber\\
    &=O\Biggr(\frac{D_v}{D^3} \Biggr(\sum_{j\ne n} \frac{1}{E_n - E_j} \Biggr)^2 \left(\frac{{\rm Tr}(h_h)^2}{D_v{\rm Tr}(h_h^2)} + 1\right)\|H_k\|_\infty^2 \Biggr)\label{eq:term1}
\end{align}
Now under the assumption that 
$$
\mathbb{E}\left(\Biggr(\sum_{j\ne n} \frac{1}{E_n - E_j} \Biggr)^2 \right)\in {O}\left({\Gamma^2D^2} \right),
$$
it follows that this term is asymptotically bounded above by $O(\Gamma^2\|H_k\|_\infty^2 /D_h)$.  It is straight forward to verify that the same bound holds for all remaining $4$ products in the expansion.

It then follows from~\eqref{eq:diff} and~\eqref{eq:term1} that
\begin{align}
    \left|\partial_{\theta_k}( \mathbb{E}(|{\rm Tr}( (O_{\rm obj}\otimes I_h)\rho(\theta_k))))\right|&=\left|\lim_{\theta_k \rightarrow 0}  \frac{\mathbb{E}(|{\rm Tr}( (O_{\rm obj}\otimes I_h)\rho(\theta_k)) - {\rm Tr}( (O_{\rm obj}\otimes I_h)\rho)|)}{\theta_k}\right|\nonumber\\
    &\in {O}\left( \|O_{\rm obj}\|_\infty\Gamma\|H_k\|_\infty \sqrt{\frac{D_v}{D_h}\left(\frac{{\rm Tr}(h_h)^2}{D_v{\rm Tr}(h_h^2)} + 1\right)}\right)
\end{align}

The claim that this bound on the derivative holds with high probability over the ensemble is then a direct consequence of Markov's inequality.
\end{proof}

\end{document}